\documentclass[12pt]{article}

\usepackage{url}
\usepackage{fullpage}
\usepackage{amssymb,amsfonts,amsmath,amsthm}
\usepackage[colorlinks=true,linkcolor=blue,citecolor=red]{hyperref}
\usepackage{algorithm}
\usepackage{algpseudocode}

\newcommand{\ov}{\overline}
\newcommand{\floor}[1]{{\lfloor#1\rfloor}}
\newcommand{\ceil}[1]{{\lceil#1\rceil}}

\newcommand{\dist}{{\sf distance}}

\renewcommand{\P}{\ensuremath{\mathcal{P}}}

\newcommand{\E}{{\mathbb E}}

\newcommand{\Var}{\textrm{Var}}
\newcommand{\eps}{\epsilon}
\newcommand{\seq}{\subseteq}

\renewcommand{\int}{{\sf int}}

\newtheorem{theorem}{Theorem}[section]

\newtheorem{proposition}[theorem]{Proposition}
\newtheorem{definition}[theorem]{Definition}

\newtheorem{claim}[theorem]{Claim}

\newenvironment{proof-sketch}{\noindent{\bf Sketch of Proof}\hspace*{1em}}{\qed\bigskip}
\newenvironment{proof-idea}{\noindent{\bf Proof Idea}\hspace*{1em}}{\qed\bigskip}
\newenvironment{proof-of-lemma}[1]{\noindent{\bf Proof of Lemma #1}\hspace*{1em}}{\qed\bigskip}
\newenvironment{proof-of-claim}[1]{\noindent{\bf Proof of Claim #1}\hspace*{1em}}{\qed\bigskip}
\newenvironment{proof-of-thm}[1]{\noindent{\bf Proof of Theorem #1}\hspace*{1em}}{\qed\bigskip}
\newenvironment{proof-attempt}{\noindent{\bf Proof Attempt.}\hspace*{1em}}{\qed\bigskip}

\renewcommand{\leq}{\leqslant}

\renewcommand{\geq}{\geqslant}
\renewcommand{\epsilon}{\varepsilon}

\title{An $\widetilde{O}(n)$ Queries Adaptive Tester for Unateness}

\author{
	Subhash Khot
	\thanks{
    Courant Institute of Mathematical Sciences,
    New York University.
    Research supported by NSF grants CCF 1422159, 1061938, 0832795  and Simons Collaboration on Algorithms and Geometry grant.
    }
	\and
	Igor Shinkar
	\thanks{
    Courant Institute of Mathematical Sciences,
    New York University. Same funding as Subhash Khot.
    }
}

\begin{document}

\maketitle
\thispagestyle{empty}

\begin{abstract}
    We present an adaptive tester for the unateness property of Boolean functions.
    Given a function $f:\{0,1\}^n \to \{0,1\}$ the tester makes $O(n \log(n)/\eps)$
    adaptive queries to the function. The tester always accepts a unate function,
    and rejects with probability at least 0.9 if a function is $\eps$-far from being unate.
\end{abstract}

\section{Introduction}\label{sec:intro}

A Boolean function $f:\{0,1\}^n \to \{0,1\}$ is said to be \emph{unate}
if for every $i \in [n]$ it is either the case that $f$ is monotone
non-increasing in the $i$'th coordinate, or $f$ is monotone
non-decreasing in the $i$'th coordinate.
In this work we present an adaptive tester for the unateness property
that makes $O(n \log(n)/\eps)$ adaptive queries to a given function.
The tester always accepts a unate function,
and rejects with probability at least 0.9 any function that is $\eps$-far from being unate.

Testing unateness has been studied first in the paper of Goldreich et al.~\cite{GGLRS},
where the authors present a non-adaptive tester for unateness that
makes $O(n^{1.5}/\eps)$ queries. The tester in~\cite{GGLRS} is the so-called ``edge tester'',
that works by querying the function on the endpoints of $O(n^{1.5}/\eps)$ uniformly random edges of the hypercube,
i.e., uniformly random pairs $(x,y)$ that differ in one coordinate,
and checking that there are no violations to the unateness property.

The notion of unateness generalizes the notion of monotonicity.
Recall that a Boolean function $f:\{0,1\}^n \to \{0,1\}$ is said to be monotone
if $f(x) \leq f(y)$ for all $x \prec y$,
where $\prec$ denotes the natural partial order on Boolean strings,
namely, $x \prec y$ if $x_i \leq y_i$ for all $i \in [n]$.
Since the original paper of~\cite{GGLRS} there has been a lot of research
concerning the problem of testing monotonicity
of Boolean functions, as well as many closely related
problems, such as testing monotonicity on functions with different (non-Boolean)
domains~\cite{DGLRRS99, FLNRRS, BCGSM12, LehmanRon, CS13, CST14, CDST15, BelovsBlais, BGJRW},
culminating in a recent result of~\cite{KMS15}, which gives a $\widetilde{O}(\sqrt{n}/\eps^2)$-query
non-adaptive tester for monotonicity. In this paper we
will use the monotonicity tester of~\cite{GGLRS}, which has a better dependence on $\eps$.

\begin{theorem}[Testing Monotonicity~\cite{GGLRS}]\label{thm:GGLRS}
	For any proximity parameter $\eps>0$
    there exists a non-adaptive tester for the monotonicity property that
    given a function $f:\{0,1\}^n \to \{0,1\}$ the tester makes $O(n/\eps)$
    queries to the function. The tester always accepts a monotone function,
    and if a function is $\eps$-far from being monotone,
    the tester finds a violation to monotonicity with probability at least $0.99$.
\end{theorem}

We remark that the monotonicity testers analyzed in~\cite{GGLRS, CS13, CST14, KMS15} are all pair testers
that pick pairs $x \prec y$ according to some distribution, and check that
the given function $f$ does not violate monotonicity on this pair, i.e.,
checks that $f(x) \leq f(y)$. It is not clear whether a variant of such tester can be
applied for testing unateness, since the function can be monotone increasing in
some of the coordinates where $x$ and $y$ differ, and monotone decreasing in others.

\subsection{Our result}\label{sec:results}

In this paper we prove the following theorem.

\begin{theorem}\label{thm:unateness testing}
	For any proximity parameter $\eps>0$
    there exists an adaptive tester for the unateness property,
    that given a Boolean function $f:\{0,1\}^n \to \{0,1\}$
    makes $O(n \log(n)/\eps)$ adaptive queries to $f$.
    The tester always accepts a unate function,
    and rejects with probability at least 0.9 any function that is $\eps$-far from being unate.
\end{theorem}

The tester works as follows.
Given a function $f:\{0,1\}^n \to \{0,1\}$,
the tester first finds a subset of coordinates $T \seq [n]$ such that
the function is essentially independent of the coordinates outside $T$.
Specifically, it finds a subset of coordinates $T \seq [n]$ such that
$\E_{z \in \{0,1\}^T} [ \Var_{w \in \{0,1\}^{[n]\setminus T}}[ f(z_T \circ w_{\ov{T}}) ]]$
is small, i.e., if we pick $x,y \in \{0,1\}^n$ that are equal
on their coordinates in $T$ uniformly at random,
then with high probability we will have $f(x) = f(y)$.
Furthermore, for each $i \in T$ the tester will find an edge $(x,x+e_i)$ in
the hypercube such that $f(x) \neq f(x+e_i)$
(where $e_i$ is the unit vector with $1$ in the $i$'th coordinate)
Querying $f$ on these two points gives a ``direction''
for monotonicity for each coordinate in $T$.

In the second part of the tester, we define a function
that depends only on the coordinates in $T$ by fixing the variables
outside $T$ uniformly at random.
We then apply the monotonicity tester from Theorem~\ref{thm:GGLRS}
on this function with respect to the directions obtained
for the coordinates in $T$ in the previous step,
and output the answer of this tester.
For the analysis, we use the fact that \emph{on average} the restricted function is close to the original function $f$,
and hence is far from being unate. In particular, it is far from being a monotone function
with respect to the directions for the coordinates in $T$ obtained in the first step.
Hence a monotonicity tester with high probability will find a violation of monotonicity
in these directions, which will serve as evidence that the function is not unate.

\section{Preliminaries}\label{sec:prelim}

\begin{definition}
For two Boolean functions $f,g:\{0,1\}^n \to \{0,1\}$
defined the distance between them as
$\dist(f,g) = \Pr_{x \in \{0,1\}^n}[f(x) \neq g(x)] = 2^{-n} |\{ x \in \{0,1\}^n : f(x) \neq g(x) \}|$.
We say that $f$ is $\eps$-far from a collection of functions $\P$ if for any $g \in \P$
it holds that $\dist(f,g) \geq \eps$.
\end{definition}

\begin{definition}
A Boolean function $f:\{0,1\}^n \to \{0,1\}$ is said to be \emph{monotone non-decreasing}
or simply \emph{monotone} if $f(x) \leq f(y)$ for all $x \prec y$,
where $\prec$ denotes the natural partial order on Boolean strings
i.e., $x \prec y$ if $x_i \leq y_i$ for all $i \in [n]$.
In other words, $f$ is monotone if for every $i \in [n]$
the function $f$ is monotone non-decreasing in the $i$'th coordinate.

For directions $B = (b_i \in \{up,down\} : i \in [n])$
let the partial order $\prec_B$ be defined as
$x \prec_B y$ if for all $i \in [n]$ such that $b_i = up$
it holds that $x_i \leq y_i$
and for all for all $i \in [n]$ such that $b_i = down$
it holds that $x_i \geq y_i$.
A Boolean function $f:\{0,1\}^n \to \{0,1\}$ is said to be \emph{monotone
with respect to the directions $B = \{b_i \in \{up,down\} : i \in [n]\}$}
if $f(x) \leq f(y)$ for all $x \prec_B y$.

A Boolean function $f:\{0,1\}^n \to \{0,1\}$ is said to be \emph{unate}
if it is monotone with respect to some directions, i.e.,
if for every $i \in [n]$ it is either the case that $f$ is monotone
non-increasing in the $i$'th coordinate, or $f$ is monotone
non-decreasing in the $i$'th coordinate.
\end{definition}

Next we make definitions related to restrictions of Boolean functions
by fixing some of the coordinates.

\begin{definition}
Given a string $x \in \{0,1\}^n$ and a subset of coordinates $T \seq [n]$
denote by $x_T$ the substring of $x$ whose coordinates are indexed by $T$.
Given two strings $x,y \in \{0,1\}^n$ and two disjoint subsets of coordinates $S,T \seq [n]$
denote by $x_T \circ y_S$ the string $z$ whose coordinates are indexed by $T \cup S$
with $z_i = x_i$ if $i \in T$ and $z_i = y_i$ if $i \in S$.
\end{definition}

\begin{definition}
Let $f:\{0,1\}^n \to \{0,1\}$ be a Boolean function.
For a subset of coordinates $T \seq [n]$ and $w \in \{0,1\}^{[n]\setminus T}$
denote by $f_{T,w} :\{0,1\}^n \to \{0,1\}$ the Boolean function defined as
$f_{T,w}(z) = f(z_T \circ w_{[n]\setminus T})$.
That is, for each $w \in \{0,1\}^{[n]\setminus T}$
the function $f_{T,w}$ depends only on the coordinates in $T$.
\end{definition}

\begin{definition}
Let $f:\{0,1\}^n \to \{0,1\}$ be a Boolean function,
and let $T \seq [n]$ be a subset of coordinates.
Define
$\Var_{[n]\setminus T}(f) = \E_{z \in \{0,1\}^T} [ \Var_{w \in \{0,1\}^{[n]\setminus T}}[ f(z_T \circ w_{\ov{T}}) ]]$.
\end{definition}

This quantity has been used before, e.g., in~\cite{KS,Blais09}.
It measures how much $f$ is depends on the coordinates outside $T$.
In particular, if $f$ depends only on the coordinates in $T$,
(i.e., is independent of the coordinates in $[n] \setminus T$)
then $\Var_{[n]\setminus T}(f) = 0$.

The following proposition is straightforward from the definition.
\begin{proposition}\label{prop:var}
Let $f:\{0,1\}^n \to \{0,1\}$ be a Boolean function.
and let $T \seq [n]$ be a subset of coordinates.
Pick $x,y \in_R \{0,1\}^n$ such that $x_i = y_i$
for all $i \in T$ and $\{ x_i,y_i \in \{0,1\} : i \in [n] \setminus T\}$ are chosen independently and uniformly at random.
Then $\Var_{[n]\setminus T}(f) = \Pr[f(x) \neq f(y)]$.
\end{proposition}

\section{Proof of Theorem~\ref{thm:unateness testing}}\label{sec:thm:unateness testing proof}

Below we present our tester for the unateness property.
The tester uses a subroutine called \textsf{Find an influential coordinate}
which works as follows.
It gets an oracle access to a Boolean function $f:\{0,1\}^n \to \{0,1\}$,
and a subset of the coordinates $T \seq [n]$, which is given explicitly.
The subroutine outputs either $\bot$
or some $i^* \in [n]\setminus T$ and $b \in \{up,down\}$ such that
there exist $x,y \in \{0,1\}^n$ that differ only in the $i^*$'th coordinate,
satisfy $f(x) \neq f(y)$, and $b$ is the orientation of $f$ along the edge $(x,y)$.

The subroutine \textsf{Find an influential coordinate} has the
guarantee that if $f$ has some non-negligible dependence on the coordinates
outside $T$, then with some non-negligible probability the subroutine will return
some $i^* \in [n]\setminus T$ and $b \in \{up,down\}$ as above.
This is done by picking independently and uniformly at random two inputs
$x,y \in \{0,1\}^n$ that are equal on their coordinates in $T$
such that $f(x) \neq f(y)$, and then using ``binary search''
in order to decrease $\dist(x,y)$ to 1, while preserving the invariant that
$f(x) \neq f(y)$.
Specifically, given $x,y \in \{0,1\}^n$ such that $f(x) \neq f(y)$
we pick an arbitrary $z \in \{0,1\}^n$ such that
if $V =\{i \in [n] : x_i \neq y_i\}$ is the set of the coordinates where $x_i=y_i$,
then $z_i = x_i$ for all $i \in [n] \setminus V$,
and $\dist(z,x) = \floor{|V|/2}$ and $\dist(z,y) = \ceil{|V|/2}$.
Since $f(x) \neq f(y)$, it must be the case that $f(z)$ differs from either $f(x)$
or $f(y)$. We then update either $x$ or $y$ to be $z$ so that $f(x) \neq f(y)$.
This clearly decreases $\dist(x,y)$ by roughly a multiplicative factor of 2,
and so, by repeating this at most $\log(n)$ times we obtain
$x$ and $y$ that satisfy $f(x) \neq f(y)$ and differ in exactly one coordinate.

\begin{algorithm}
  \begin{algorithmic}[1]
	\label{alg:find direction}
	\Procedure{Find an influential coordinate}{$f:\{0,1\}^n \to \{0,1\}, T$}
    \State Pick $x,y \in_R \{0,1\}^n$ independently and uniformly at random such that $x_T = y_T$\label{line:pick xy}
	\If{$f(x) = f(y)$}
    \State
	\Return $\bot$
	\Else {  $( f(x) \neq f(y) )$}
		\Repeat
			\State $U \gets \{i \in [n] : x_i = y_i\}$
			\State $V \gets \{j \in [n] : x_j \neq y_j\}$
			\State Pick an arbitrary $z_V \in \{0,1\}^V$ such that $|\{i \in V : z_i = y_i\}| = \floor{|V|/2}$.
			\State Let $z = x_U \circ z_V \in \{0,1\}^n$
			\If {$f(x) \neq f(z)$}
				\State $y \gets z$
			\Else {  $( f(y) \neq f(z) )$}
				\State $x \gets z$
			\EndIf
		\Until{$|V| = 1$}
		\State	Let $i^* \in [n]$ be the unique element in $V$
		\State	Let $b \in \{up,down\}$ be the orientation of $f$ in the edge $(x,y)$
	    \State	\Return $(i^*, b)$
    \EndIf
    \EndProcedure
  \end{algorithmic}
\end{algorithm}

\begin{algorithm}\label{alg:test}
  \begin{algorithmic}[1]
    \Procedure{Unateness tester}{$f:\{0,1\}^n \to \{0,1\}$}
      \State Let $m = O(\frac{n}{\eps})$
      \State Let $T = \emptyset$
      \For {$i =1...m$}
        \State \textsf{Find an influential coordinate}($f, T$)
        \If {returned a coordinate and a direction $(i^*,b_{i^*})$}
        \State Add $i^*$ to $T$, and let $b_{i^*}$ be the corresponding direction.
        \EndIf
      \EndFor

      \State Pick $w \in \{0,1\}^{[n]\setminus T}$
      \State Apply the monotonicity tester
            on $f_{T,w}$ with respect to the directions $\{b_i: i \in T\}$\label{line:mon test}
      \State Return the output of the monotonicity tester.
    \EndProcedure
  \end{algorithmic}
\end{algorithm}

For the proof of Theorem~\ref{thm:unateness testing} we need the following two claims.

\begin{claim}\label{claim:1}
	Let $c>0$ be a small constant and let $m = \frac{2n}{c\eps}$
	be the number of iterations of the {\em for} loop
	in the Unateness tester.
	Let $f: \{0,1\}^n  \to \{0,1\}$ be a Boolean function,
	and let $T \seq [n]$ be the set in the Unateness tester
	after $m$ iterations of the {\em for} loop.
	Then, with high probability the set $T$ satisfies
	\[
		\Var_{[n]\setminus T}(f) < c\eps.
	\]
\end{claim}

\begin{proof}
	Note that if in some iteration we have a subset of coordinates $T \seq [n]$
	such that $\Var_{[n]\setminus T}(f) \geq c\eps$,
	then, by Proposition~\ref{prop:var} the variables $x$ and $y$
	chosen in line~\ref{line:pick xy} of \textsf{Find an influential coordinate($f, T$)}
	will satisfy $f(x) \neq f(y)$ with probability at least $c\eps$.
	Having such $x$ and $y$, let $U \seq [n]$ be the coordinates where $x$ and $y$
	are equal, and let $V \seq [n]$ be the coordinates where the two strings
	differ. Then, in each iteration the procedure chooses $z$ at random, such that
	it agrees with $x$ and $y$ in the coordinates where they equal, and updates
	$x$ or $y$ according to the value of $f(z)$, while preserving the property
	that $f(x) \neq f(y)$. Clearly, if $z \neq y$ and $z \neq x$, then in each
	step we reduce the distance between $x$ and $y$, until $|V| = 1$,
	i.e., $y = x + e_i$ for the unique coordinate $i \in V$,
	which is returned by the procedure together with the orientation of the edge $(x,y)$.

	Therefore, if $m = \frac{2n}{c\eps}$,
	then by Azuma's inequality with probability $1 - e^{-\Omega(n)}$
	among the $m$ iterations
	at least $\frac{c\eps m}{2} = n$ iterations will have the property that either
	\textsf{Find an influential coordinate} finds a new coordinate to add to $T$
	or that $\Var_{[n]\setminus T}(f) \leq c\eps$.%
	\footnote{Formally, let $(X_i: i \in [m])$ be Bernouli random variables
	with $X_i=1$ if either $\Var_{[n]\setminus T}(f) \leq c\eps$
	or a new coordinate is added to $T$ in the $i$'th iteration,
	and observe that $\Pr[X_i=1] \geq c\eps$ for all $i \in [m]$.}
	On the other hand, the function $f$ depends on at most $n$ coordinates,
	and hence, after $m = \frac{2n}{c\eps}$ iterations
	the set $T$ will satisfy the property
	\[
		\Var_{[n]\setminus T}(f) \leq c\eps,
	\]
	with probability at least $1 - e^{-\Omega(n)}$, as required.
\end{proof}

\begin{claim}\label{claim:2}
	Let $f: \{0,1\}^n  \to \{0,1\}$ be a Boolean function,
	and let $T \seq [n]$ be such that
	\[
		\Var_{[n]\setminus T}(f) \leq c\eps
	\]
	for some $\eps>0$ and $c \in (0,1/8)$.
	Then, for a random $w \in \{0,1\}^{[n]\setminus T}$ it holds that
	\[
	\Pr_{w \in \{0,1\}^{[n]\setminus T}}[ \dist(f_{T,w},f) \geq \eps/2 ] \leq 8c.
	\]
\end{claim}

\begin{proof}
	Define the function $Maj_T: \{0,1\}^n \to \{0,1\}$ as
	\[
		Maj_T(z) = \begin{cases}
						1 & \text{if } \Pr_{w \in \{0,1\}^{[n]\setminus T}}[f(z_T \circ w_{\ov{T}}) = 1] > 0.5 \\
						0 & \text{otherwise.}
					\end{cases}
	\]
	That is, $Maj_T$ depends only on the coordinates in $T$.
	By the assumption of the claim we have that for a uniformly random $w \in \{0,1\}^{[n]\setminus T}$
	it holds that
	\begin{eqnarray*}
	\E_{w \in \{0,1\}^{[n]\setminus T}}[\dist(f_{T,w}, Maj_T)]
	& = & \E_{z \in \{0,1\}^T}[ \Pr_{w \in \{0,1\}^{[n]\setminus T}} [f(z_T \circ w_{\ov{T}}) \neq Maj(z_T)] \\
	& \leq &\E_{z \in \{0,1\}^T} [ 2 \Var_{w \in \{0,1\}^{[n]\setminus T}}[ f(z_T \circ w_{\ov{T}}) ] ] \\
	& \leq & 2c\eps.
	\end{eqnarray*}
	Hence, by Markov's inequality
	\[
	\Pr_w[\dist(f_{T,w}, Maj_T) \geq \eps/4 ] \leq 8c.
	\]
	On the other hand,
	\[
		\dist(f,Maj_T) =	\E_{w \in \{0,1\}^{[n]\setminus T}}[\dist(f_{T,w}, Maj_T)] \leq 2c\eps \leq \eps/4.
	\]
	Therefore, by triangle inequality we have
	\[
	\Pr_w[\dist(f_{T,w}, f) \geq \eps/2]
	\leq
	\Pr_w[\dist(f_{T,w}, Maj_T) \geq \eps/4]
	\leq 8c,
	\]
	and the claim follows.
\end{proof}

Theorem~\ref{thm:unateness testing} now follows easily from the above claims.
\begin{proof}[Proof of Theorem~\ref{thm:unateness testing}]
	
	For a small constant $c>0$ let $m = O(\frac{n}{c\eps})$
	be the number of iterations of the {\em for} loop in the Unateness tester.
	Let $T \seq [n]$ be the set in the Unateness tester
	after $m$ iterations of the \textsf{for} loop.
    By Claim~\ref{claim:1} with probability $0.99$ the set $T$ satisfies
	\[
		\Var_{[n]\setminus T}(f) \leq c\eps.
	\]
    Assuming that $T$ satisfies the above,
	by Claim~\ref{claim:2} if $f$ is $\eps$-far from being unate,
	then for a uniformly random	$w \in \{0,1\}^{[n]\setminus T}$ it holds that
	$f_{T,w}$ is $\eps/2$-far from being unate with probability $(1-8c)$,
	and in particular, it is $\eps/2$-far from being monotone with respect
    to the directions $\{b_i: i \in T\}$.
    By applying the monotonicity tester on $f_{T,w}$ with $w$ such that
    $f_{T,w}$ is $\eps/2$-far from being unate it follows that
    with probability at least $0.99$ the invocation of the monotonicity tester
    will find a violation to monotonicity of $f_{T,w}$ with respect to the directions $\{b_i: i \in T\}$.
	Therefore, for a sufficiently small constant $c>0$,
	if $f$ is $\eps$-far from unate, then with probability $0.9$ the tester will reject.

	Finally, we analyze the query complexity of the tester.
	It is clear that the procedure
    \textsf{Find an influential coordinate} makes at most $O(\log(n))$ iterations,
    as in each iteration $z$ differs from both $x$ and $y$ in at most $\ceil{\dist(x,y)/2}$ coordinates.
	Therefore, the total number of queries made by the tester in the
	\textsf{for} loop is $m \cdot O(\log(n))$. In addition the tester makes at most $O(n/\eps)$
	queries in step~\ref{line:mon test}.
	Therefore, the total number of queries made by the tester is at most $O(n \log(n)/\eps)$.
\end{proof}

\section*{Acknowledgements}
We thank the anonymous referees for their helpful comments.

\bibliographystyle{alpha}
\bibliography{unate}

\end{document}